\documentclass[letterpaper,final,authoryear,3p,11pt]{elsarticle} 
\pdfoutput=1
\usepackage{amssymb,amsthm}

\bibpunct{(}{)}{,}{a}{}{,}
\usepackage{amsmath,bm,url}
\newcommand{\vc}[1]{\bm{#1}}
\newcommand{\mat}[1]{{\mathbf #1}}  
\newcommand{\der}[2]{\frac{{\mathrm d}#1}{{\mathrm d}#2}}
\newcommand{\pder}[2]{\frac{\partial #1}{\partial #2}}
\newcommand{\dr}[2]{{\mathrm d}#1/{\mathrm d}#2}
\newcommand{\pdr}[2]{\partial #1/\partial #2}
\newcommand{\dd}{\,{\mathrm d}}

\newcommand{\abs}{\mathop{\mathgroup\symoperators abs}\nolimits}
\providecommand{\e}{\mathrm e} 
\newcommand{\E}{\mathop{\mathgroup\symoperators E}\nolimits}

\newcommand{\Var}{\mathop{\mathgroup\symoperators Var}\nolimits}
\newcommand{\Cov}{\mathop{\mathgroup\symoperators Cov}\nolimits}
\newtheorem{proposition}{Proposition}
\newcommand{\fig}[3][]{\begin{figure}[htbp]\leavevmode\centering%
\includegraphics[width=#1\textwidth]{#2.pdf}\caption{\small #3}\label{fig:#2}\end{figure}}
  \makeatletter
  \def\ps@pprintTitle{%
     \let\@oddhead\@empty
     \let\@evenhead\@empty
     \def\@oddfoot{\footnotesize\itshape
       \@journal}%
     \let\@evenfoot\@oddfoot}
  \makeatother

\journal{Manuscript in review \hfill May 2018}

\begin{document}

\begin{frontmatter}

\title{Estimating reducible stochastic differential equations by\\ conversion to a least-squares problem}

\author{Oscar Garc\'{\i}a}

\address{Dasometrics, Conc\'on, Chile}
\ead{garcia@dasometrics.net}
\fntext[]{ORCID: 0000-0002-8995-1341}

\begin{abstract}
Stochastic differential equations (SDEs) are increasingly used in longitudinal data analysis, compartmental models, growth modelling, and other applications in a number of disciplines. Parameter estimation, however, currently requires specialized software packages that can be difficult to use and understand. This work develops and demonstrates an approach for estimating reducible SDEs using standard nonlinear least squares or mixed-effects software. Reducible SDEs are obtained through a change of variables in linear SDEs, and are sufficiently flexible for modelling many situations.
The approach is based on extending a known technique that converts maximum likelihood estimation for a Gaussian model with a nonlinear transformation of the dependent variable into an equivalent least-squares problem. A similar idea can be used for Bayesian maximum a posteriori estimation.
It is shown how to obtain parameter estimates for reducible SDEs containing both process and observation noise, including hierarchical models with either fixed or random group parameters. Code and examples in \textit{R} are given. Univariate SDEs are discussed in detail, with extensions to the multivariate case outlined more briefly.
The use of well tested and familiar standard software should make SDE modelling more transparent and accessible.
\end{abstract}

\begin{keyword}
stochastic processes \sep longitudinal data \sep growth curves \sep compartmental models \sep mixed-effects \sep \textit{R}
\end{keyword}

\end{frontmatter}

\setcounter{tocdepth}{4}
\tableofcontents

\section{Introduction}
\label{sec:intro}

Stochastic differential equation (SDE) models incorporate random elements into ordinary differential equations, representing the effects of a noisy environment on rates of change. In addition to their traditional applications in physics and engineering \citep{gardiner94}, SDEs are becoming routine in finance \citep[][Sec.~7.3]{kunita10,kloeden92} and pharmacokinetics \citep{kristensen05,klim09,donnet13}, and are increasingly used in fields like econometrics \citep{paige10,bu10,bu16}, animal growth \citep{lv07,strathe09,filipe10}, oncology \citep{sen89,favetto10}, and forestry \citep{garcia83,broad06,lodgepole}. Other applications include \citet{artzrouni90,cleur00,driver17}. In particular, deterministic growth curves are commonly fitted to repeated measurements assuming that observation error is the only source of randomness \citep{davidian03}, but an SDE growth model can be more realistic \citep{hotelling,donnet10}. In addition to environmental or process noise, some models include also observation errors \citep{garcia83,donnet10,donnet13}. Another extension of the basic SDE is to hierarchical models with repeated measurements, where some parameters may be specific to each individual or group. Such \emph{local} parameters can be fixed numbers \citep{garcia83} or, more often, are viewed as random  in mixed-effects models \citep{donnet08,klim09,picchini10,driver17}.

Parameter estimation for general SDEs is a difficult problem, with most algorithms employing approximations to maximum likelihood \citep[e.~g.,][]{nielsen00,bishwal08,picchini10,donnet13}. Bayesian methods have also been proposed \citep{donnet10,brouste14,king16,whitaker17}. \textit{R} packages implementing estimation for various types of SDE models include \textit{ctsem}, \textit{dynr}, \textit{mixedsde}, \textit{msde}, \textit{pomp}, \textit{sde}, \textit{Sim.DiffProc}, and \textit{yuima}  \citep{cran}; \textit{pomp} and \textit{yuima} include Bayesian procedures. An alternative method is presented here for SDEs that can be reduced to linear by a change of variables. The approach is based on an extension of ideas of \citet{box-cox} and \citet{furnival61}, which convert maximum-likelihood (ML) estimation for a transformed Gaussian model into a least-squares problem. The same principle can be used to obtain Bayesian maximum a posteriori (MAP) estimates. The sum of squares can be minimized with standard statistical software like the \texttt{nls} function in \textit{R}, or \texttt{nlme} for mixed-effects formulations \citep{pinheiro00}.

The general methodology is described in the next section, followed by its application to SDE models with observation errors in Section 3. Examples in \textit{R} are given. The SDE models may include hierarchical structures with group fixed or random effects (Sec.~\ref{sec:hier}). To simplify, only univariate SDEs are discussed in full detail, with extensions to vector-valued multiresponse models treated more briefly in Section \ref{sec:multi}. A Conclusions section closes the article.

\section{Transformed Gaussian model}
\label{sec:tgm}

\subsection{Example 1}
\label{sec:ex1}

\citet[][Sec.~8.7]{venables02} describe a dataset containing measurements of the concentration of a chemical GAG in the urine of 314 children of various ages. Assume that one wants an equation summarizing the relationship between GAG concentration (variable $y$) and age (variable $x$). Plotting shows that the relationship is nonlinear, so that a simple linear regression $y = \beta_0 + \beta_1 x$ is unsatisfactory. A common strategy is to try regressions on transformations of $x$ and/or $y$, such as logarithms or various positive or negative powers. We shall find ``good'' exponents for powers in both $x$ and $y$.

Instead of a simple power transformation, it is slightly more convenient to use
\begin{equation} \label{eq:bc}
  y^{(\lambda)} = \begin{cases}
    \frac{y^\lambda - 1}{\lambda} & \text{if } \lambda \neq 0 \\
    \ln y & \text{if } \lambda = 0
    \end{cases}
\end{equation}
\citep{box-cox}. This is an increasing function of $y$ for any $\lambda$, it is continuous in $\lambda$ at $\lambda = 0$, and includes the logarithmic transformation as a special case. Physicists call eq.~\eqref{eq:bc} a \emph{generalized logarithm} \citep{martinez08}.

The model is then
\begin{equation} \label{eq:ex1}
  y_i^{(\lambda_y)} = \beta_0 + \beta_1 x_i^{(\lambda_x)} + \epsilon_i \;.
\end{equation}
Assuming that the $\epsilon_i$ are independent and normally distributed, it is possible to write down the likelihood function and use a general optimization algorithm to find the maximum likelihood (ML) estimates for the parameters $\beta_0, \beta_1, \lambda_x, \lambda_y$, and $\sigma^2 = \Var[\epsilon_i]$. However, the optimization can be time-consuming and ill-conditioned, requiring good starting points to achieve convergence. \citet{box-cox} showed a way of converting the likelihod maximization into a least-squares solution, which is more efficient and stable. \citet{furnival61} had used the same idea to devise an index for comparing dependent variable transformations. Specifically, \citet{box-cox} considered various one- and two-parameter transformations on the left-hand side and a linear right-hand side, solving by ordinary linear least squares. Estimation in eq.~\eqref{eq:ex1} is similar, except that with a free $\lambda_x$ it requires a non-linear least-squares procedure.

The next two sections develop an extension of the Box--Cox--Furnival approach, needed for the application to SDEs. The ML estimation for Example 1, both by the method of \citet{box-cox} and by direct maximization of the likelihood, is presented in Section \ref{sec:sol1}.

\subsection{Generalized formulation}
\label{sec:gen}

Consider a vector of $n$ observations $\vc y = (y_1, \ldots, y_n)$, and a transformation
\begin{equation*}
   \vc z = \vc\varphi(\vc y, \vc \theta_1)
\end{equation*}
that may depend on a vector of unknown parameters $\vc\theta_1$. Assume that $\vc\varphi$ is one-to-one and differentiable over the admissible domain, and that the transformed variables follow a Gaussian nonlinear model
\begin{equation*}
  z_i = f_i(\vc\theta_2) + \epsilon_i \;,
\end{equation*}
where the $\epsilon_i$ are independent normally distributed with mean 0 and a common variance $\sigma^2$. In \citet{box-cox} and in Section \ref{sec:ex1} the transformations are of the simpler form $z_i = \varphi(y_i, \vc\theta_1)$, $f_i$ takes the form $f(\vc x_i, \vc\theta_2)$ where $\vc x_i$ is a vector of predictors, and the parameters in $\vc\theta_1$ and in $\vc\theta_2$ are different. Those restrictions are not imposed here. In particular, we denote by $\vc\theta$ the vector containing the union of the elements from $\vc\theta_1$ and $\vc\theta_2$, some of which may be the same, and write
\begin{equation} \label{eq:trans}
   \vc z = \vc\varphi(\vc y, \vc \theta)
\end{equation}
and
\begin{equation} \label{eq:normal}
  z_i = f_i(\vc\theta) + \epsilon_i \;,
\end{equation}

There is no loss in generality from the assumption of homoscedastic uncorrelated errors in eq.~\eqref{eq:normal}, because any other correlation structure can be accommodated through a linear transformation (Sec.~\ref{sec:corr}).

\subsection{Maximum likelihood}
\label{sec:ml}

The following extends the results of \citet{box-cox} and \citet{furnival61} to the more general formulation of eqns.~\eqref{eq:trans}--\eqref{eq:normal}:
\begin{proposition}
 Let
\begin{equation} \label{eq:J}
  J = J(\vc y, \vc\theta) = \abs\left(\left\lvert\pder{\vc \varphi}{\vc y}\right\rvert\right)
  = \abs\left(\left\lvert\left\{\pder{\varphi_i}{y_j}\right\}\right\rvert\right)
\end{equation}
be the absolute value of the Jacobian determinant of the transformation \eqref{eq:trans}, and write
\begin{equation} \label{eq:u}
  u_i = \frac{\epsilon_i}{J^{1/n}} =
    \frac{\varphi_i(\vc y, \vc \theta) - f_i(\vc\theta)}{J(\vc y, \vc\theta)^{1/n}} \;.
\end{equation}
Then, the ML estimate of $\vc\theta$ is the value $\hat{\vc\theta}$ that minimizes the sum of squares $S(\vc\theta) = \sum u_i^2$. The ML estimate of $\sigma^2$ is
\begin{equation} \label{eq:s2hat}
  \hat{\sigma^2} = J(\vc y, \hat{\vc\theta})^{2/n} \frac{S(\hat{\vc\theta})}{n} \;.
\end{equation}
\end{proposition}

\begin{proof}
The probability density of $\vc y$, and hence the likelihood in terms of the original observations, is the product of the density of $\vc z$ and the absolute value of the  Jacobian of the transformation \citep[e.~g.,][Sec.~2.8]{wilks62},
\begin{equation} \label{eq:Lraw}
  L = \frac{1}{(2\pi)^{n/2} \sigma^n} \exp\left(- \frac{\sum \epsilon_i^2}{2 \sigma^2}\right) J \;.
\end{equation}
Substituting $\epsilon_i = J^{1/n} u_i$ (from eq.~\eqref{eq:u}) and $\sigma = J^{1/n} \sigma_u$,
\begin{equation} \label{eq:L}
  L = \frac{1}{(2\pi)^{n/2} \sigma_u^n} \exp\left(- \frac{\sum u_i^2}{2 \sigma_u^2}\right) \;.
\end{equation}
This has the form of a normal density, and therefore the standard normal theory applies \citep[e.~g.,][Sec.~2.2.1]{seber-wild}. Briefly, equating to 0 the derivative of $L$ with respect to $\sigma_u^2$ it is found that, for given $\vc\theta$, the likelihood is maximized by
\begin{equation} \label{eq:s2}
  \sigma_u^2 = \frac{\sum u_i^2}{n} \;.
\end{equation}
Substituting into eq.~\eqref{eq:L},
\begin{equation*}
  \max_{\sigma_u^2} L = \left(\frac{n}{2 \pi \sum u_i^2}\right)^{n/2} \exp(-n/2) \;.
\end{equation*}
Clearly, the ML estimate $\hat{\vc\theta}$ can be calculated minimizing over $\vc\theta$ the sum of squares $\sum u_i(\vc y, \vc\theta)^2$. Equation \eqref{eq:s2} and $\sigma^2 = J^{2/n} \sigma_u^2$ give eq.~\eqref{eq:s2hat}.
\end{proof}

An alternative variance estimator with $n - p$ instead of $n$ in the denominator of eq.~\eqref{eq:s2hat} is often used, which generally has less bias but a larger mean square error. It may be worth noting that in general the $u_i$ are not Gaussian, even though the likelihood takes that form.

\subsection{Solution of Example 1}
\label{sec:sol1}

Eq.~\eqref{eq:ex1} corresponds to \eqref{eq:trans}--\eqref{eq:normal} with $z_i = y_i^{(\lambda_y)}$and $\vc\theta = (\beta_0, \beta_1, \lambda_x, \lambda_y)$.

In this instance the Jacobian matrix is diagonal, and the determinant is the product of the diagonal elements $\pdr{z_i}{y_i}$. Then, as shown by \citet{furnival61} and by \citet{box-cox},
\begin{equation} \label{eq:bcJnroot}
  J^{1/n} = {\dot y}^{\lambda - 1} \;,
\end{equation}
where $\dot y$ denotes the geometric mean of the $y_i$.

\textit{R} commands for finding the solution are given in \ref{app:sol1}. First, the exponents are fixed as $\lambda_x = 1$ and $\lambda_y = 0$, and a simple linear regression $\ln y = \beta_0 + \beta_1 x$ is used to produce suitable starting points for the optimization. Then, the sum of squares of the $u_i$ from eq.~\eqref{eq:u} is minimized with the nonlinear regression procedure \texttt{nls}, making use of eq.~\eqref{eq:bcJnroot}. The ML estimates are found to be
\[ \beta_0 = 3.314 \;,\quad \beta_1 = -0.3502 \;,\quad
   \lambda_x = 0.4249 \;,\quad \lambda_y = 0.1032 \;. \]
From eq.~\eqref{eq:s2hat}, the ML variance estimate is $\sigma^2 = 0.3839$.

A more convenient and parsimonious relationship can be obtained by rounding $\lambda_x = 0.5, \lambda_y = 0$, and using the linear regression
\[
  \ln y = 3.366 - 0.5069 \sqrt{x} \;.
\]

It is also shown in \ref{app:sol1} that the same estimates are obtained by direct minimization of the negative log-likelihood from eq.~\eqref{eq:Lraw}, confirming the validity of the Box--Cox approach.

\subsection{Correlation}
\label{sec:corr}

Suppose that the $\epsilon_i$ in eq.~\eqref{eq:normal} are not independent. Let the variance-covariance matrix of the vector $\vc\epsilon = (\epsilon_1, \ldots, \epsilon_n)$ be
\begin{equation} \label{eq:cov}
  \Var(\vc\epsilon) = \sigma^2 \mat C \;,
\end{equation}
where the matrix $\mat C$ may depend on unknown parameters. The Cholesky factorization gives $\mat C$ as the product of a lower-triangular matrix $\mat L$ and its transpose:
\begin{equation} \label{eq:cholesky}
  \mat C = \mat L \mat L' \;.
\end{equation}
It follows that the elements of $\mat L^{-1} \vc \epsilon$ are uncorrelated, and so are those of $\mat L^{-1} \vc z$  \citep[e.~g.,][Sec.~2.1.4, where the upper-triangular Cholesky factor $\mat U = \mat L'$ is used]{seber-wild}. Therefore, with the substitution $\vc\varphi(\vc y, \vc\theta) \to \mat L^{-1} \vc\varphi(\vc y, \vc\theta)$ in eq.~\eqref{eq:trans} we are in the same situation as before. Eq.~\eqref{eq:J} becomes
\begin{equation} \label{eq:Jcorr}
  J = \abs\left(\left\lvert\pder{\mat L^{-1} \vc \varphi}{\vc y}\right\rvert\right)
    = \abs\left(\left\lvert\left\{\pder{\varphi_i}{y_j}\right\}\right\rvert\right)
      / \lvert\mat L\rvert \;,
\end{equation}
with
\[
  \lvert\mat L\rvert = \prod_{i=1}^n l_{ii}
\]
since the determinant of a triangular matrix is the product of the elements on the diagonal, and those elements $l_{ii}$ in $\mat L$ are positive.
In practice it is better to use logarithms to prevent numerical overflow or underflow:
\begin{equation} \label{eq:detL}
  \lvert\mat L\rvert^{1/n} = \exp\left(\tfrac{1}{n} \sum_{i=1}^n \ln l_{ii}\right) \;.
\end{equation}
With this $J$, the vector of $u_i$ for the sum of squares is now
\begin{equation} \label{eq:ucorr}
  \vc u = \frac{\mat L^{-1} \vc\epsilon}{J^{1/n}}
\end{equation}
(c.~f.~eq.~\eqref{eq:u}).

\subsection{Maximum a posteriori estimation}
\label{sec:map}

In Bayesian MAP estimation the likelihood is weighted by an \emph{a priori} parameter distribution density:
\[
  \max_{\vc\theta, \sigma^2} L(\vc\theta, \sigma^2) p(\vc\theta) \;.
\]
If $L$ is Gaussian, this has the same form as eq.~\eqref{eq:Lraw} with $p(\vc\theta)$ in place of $J$. Therefore, analogously to Sec.~\ref{sec:ml}, the MAP estimate of $\vc\theta$ is the value $\hat{\vc\theta}$ that minimizes the sum of squares $S(\vc\theta) = \sum (\epsilon_i / p^{1/n})^2$, and the estimate of $\sigma^2$ is $\hat{\sigma^2} = p(\hat{\vc\theta})^{2/n} S(\hat{\vc\theta}) / n$. Of course, this can be combined with variable transformations as above.

\section{Stochastic differential equations}
\label{sec:sdes}

The behavior of systems that evolve in time is often modelled by a relationship betwn a magnitude $x$ and its rate of change, given by a differential equation $\dr{x}{t} = f(x)$. In many situations the rate is subject to perturbations, due to varying environmental conditions, that cause substantial deviations from the predicted trajectory. It is then more realistic to include noise, represented by a random function $\epsilon(t)$:
\begin{equation} \label{eq:langevin}
  \der{x}{t} = f(x) + g(x) \epsilon(t) \;.
\end{equation}
In particular, $\epsilon$ may be assumed to be \emph{white noise}, where the values $\epsilon(t)$ have the same distribution for all $t$, with mean 0, and values for different $t$'s are uncorrelated.
Such an SDE was proposed by Langevin in 1908 as a simpler formulation of Einstein's 1905 theory of Brownian motion. \citet{hotelling} used a similar model for logistic population growth. By the middle of the 20th Century it had been found that a rigorous mathematical treatment of the theory becomes surprisingly complex. Textbooks go over extensive advanced mathematical background material before dealing with SDEs \citep{allen07,arnold74,gardiner94,henderson06,kloeden92,oksendal03}. For modelling, an intuitive understanding may suffice.

One problem with eq.~\eqref{eq:langevin} is that a completely uncorrelated $\epsilon(t)$ would be discontinuous everywhere and have infinite variance. The way around that is to work instead with its integral $W(t)$, known as a Brownian motion or Wiener process. Increments $\Delta W$ represent the accumulated white noise over a time interval $\Delta t$, and are uncorrelated for non-overlapping intervals. Because of this, the expected value of $\Delta W$ is 0 and the variance is proportional to $\lvert\Delta t\rvert$; the proportionality constant is standardized as 1. As a sum of a potentially large number of uncorrelated increments, $W(t)$ is Gaussian. The white noise in eq.~\eqref{eq:langevin} would correspond to the limit of $\Delta W / \Delta t$ as $\Delta t$ tends to 0 but, although $W(t)$ is continuous, its fine structure is exceedingly rough (it is not differentiable), and such limit does not exist. SDEs are therefore conventionally written in terms of differentials:
\begin{equation} \label{eq:sde}
    \dd X(t) = f[X(t), t] \dd t + g[X(t), t] \dd W(t) \;,
\end{equation}
also called a continuous-time diffusion model.
The capital letters emphasize the fact that $X(t)$ and $W(t)$ are random variables. In general, $f$ and $g$ can depend on time, in a \emph{time-variant} SDE, although it is more common to assume that behavior does not change over time (\emph{time-invariant} or \emph{autonomous}). The function $f$ is sometimes called the drift or trend, and $g$ the diffusion function or volatility.

Another difficulty lies with the definition of the integral of the second term on the right-hand side of eq.~\eqref{eq:sde}. There $X$ and $W$ are stochastic processes and, unlike what happens with ordinary functions, when representing the integral as the limit of a sum the location within a time interval where $X(t)$ is evaluated makes a difference. Two definitions are commonly used, that of Ito, with $X(t)$ evaluated at the start of each interval, and that of Stratanovich, based on the interval mid-point. We will use only volatilities $g$ independent of $X$, in which case Ito and Stratonovich coincide, and the rules of ordinary calculus hold. A useful fact is that, for any sufficiently well-behaved function $g$,
\begin{equation} \label{eq:delta}
  \delta = \int g(t) \dd W(t) \text{ is Gaussian with } \E[\delta] = 0 \text{ and }
    \Var[\delta] = \int g(t)^2 \dd t \;.
\end{equation}

More generally, the state of the system may be specified by several variables, and then $X(t)$ is a vector. To simplify the presentation only one-dimensional time-invariant SDEs will be discussed in detail, mentioning briefly the main changes needed for time-variant models. Most results can be extended to multiple state variables by substituting appropriate vectors and matrices (Sec.~\ref{sec:multi}).

In general, explicit forms for the probability distributions associated to eq.~\eqref{eq:sde} are not available, and parameter estimation is difficult \citep{nielsen00,bishwal08,donnet13}. Linear SDEs, however, are more tractable and can be explicitly solved. The same is true for any process that can be reduced to a linear SDE through a sufficiently smooth nonlinear transformation \citep[][Sec.~4.3]{kloeden92}. The model used here is a linear SDE
\begin{equation} \label{eq:linear}
  \dd Y = (\beta_0 + \beta_1 Y) \dd t + \sigma_p \dd W \;,
\end{equation}
where $Y$ is a one-to-one differentiable transformation $Y = \varphi(X, \vc\lambda)$ of the original variable $X$, parameterized by a one- or higher-dimensional vector $\vc\lambda$.
The constants $\beta_0, \beta_1, \sigma_p$, and the elements in $\vc\lambda$, are parameters to be estimated. There are no additional difficulties if a parameter appears both in the transformation and in the rest of the model, so that, more generally,
\begin{equation} \label{eq:tr}
  Y = \varphi(X, \vc\theta) \;,
\end{equation}
where $\vc\theta$ is a vector including all the model parameters.
In the literature, $Y$ and $\dd W$ are often written as $Y(t)$ and $\dd W(t)$, or as $Y_t$ and $\dd W_t$, respectively.

The data are observations $x_i$ of $X(t_i)$ at $n$ occasions $t_1 < t_2 < \cdots < t_n$. Optionally, the observations may be subject to measurement or sampling errors such that
\begin{equation} \label{eq:data}
  y_i = \varphi(x_i, \vc\theta) = Y(t_i) + \epsilon_i \; ;
   \quad i = 1, \ldots, n \;,
\end{equation}
with independent normal errors $\epsilon_i \sim N(0, \sigma_m^2)$.
The initial condition $Y(t_0)$ at $t_0 < t_1$ may be a known or unknown constant, or a Gaussian random variable. In general, $Y(t_0) = y_0 + \epsilon_0$, with $\epsilon_0 \sim N(0, \sigma_0^2)$ and independent of the other $\epsilon_i$ ($\sigma_0$ may be 0).

In a time-variant model, $\beta_0, \beta_1$ and $\sigma_p$  can be functions of $t$, possibly containing unknown parameters. The transformation $\varphi$ could also depend on $t$ \citep{bu16}.

The linear SDE in eq.~\eqref{eq:linear} is known as the Ornstein-Uklenbeck process in physics, or the Vasicek model in financial economics. Parameter estimation for linear SDEs, including multivariate and mixed-efects models with observation errors, has been implemented in the \textit{R} packages \textit{PSM} by \citet{klim09} and \textit{ctsem} by \citet{driver17}. Reducible SDEs have been used by \citet{garcia83}, \citet{lv07}, \citet{paige10}, \citet{filipe10}, and \citet{bu10,bu16}.

A power transformation $Y = X^\lambda$ in eq.~\eqref{eq:tr} produces a stochastic version of the Bertalanffy-Richards growth model \citep{bertalanffy49,richards59}, which is very flexible and includes other commonly used models as special cases: von Bertalanffy's animal growth equation with $\lambda = 1/3$, the Mistcherlich or monomolecular with $\lambda = 1$, and the logistic with $\lambda = -1$. The Box-Cox transformation $Y = X^{(\lambda)}$ from eq.~\eqref{eq:bc} includes in addition the Gompertz model, $\lambda = 0$ \citep[][sections 7.3 and 7.5.3]{garcia83,seber-wild}. \citet{grex} discusses a ``double Box-Cox'' transformation, with two shape parameters, that covers nearly all the sigmoid growth curves from the literature.

In some applications the emphasis lies on modelling the volatility, as the trend may not be important if the system operates near the steady state \citep{bu10,bu16}.  A volatility function $g(X, t)$ can be converted to one that does not depend on $X$ by using the Lamperti transform
\[
    Y = \int \frac{\dd X}{g(X, t)}
\]
\citep[see][for details]{lamperti-sde}. In particular, an SDE with multiplicative error $\sigma X \dd W$ can be converted to one with additive error $\sigma \dd W$ by a logarithmic transformation.

\subsection{Estimation}
\label{sec:sdeest}

From linearity, $Y(t)$ and the $y_i$ are Gaussian. The means, variances and covariances needed for the likelihood can be obtained by integrating eq.~\eqref{eq:linear}.
Assume for now that $\beta_1 \neq 0$, the changes needed if $\beta_1 = 0$ are indicated later. As in linear ordinary differential equations, let us multiply throughout by the integrating factor $\e^{-\beta_1 t}$. Using the rule for differentiation of a product,
\begin{align*}
  \dd \left[\e^{- \beta_1 t} (Y + \beta_0 / \beta_1)\right] &= \e^{- \beta_1 t}
  \dd (Y + \beta_0 / \beta_1) + \dd(\e^{- \beta_1 t})(Y + \beta_0 / \beta_1) \\
  &= \e^{- \beta_1 t} [\dd Y - (\beta_0 + \beta_1 Y) \dd t] \;,
\end{align*}
so that
\begin{equation} \label{eq:if}
  \dd \left[\e^{- \beta_1 t} (Y + \beta_0 / \beta_1)\right] =
  \sigma_p \e^{- \beta_1 t} \dd W \;.
\end{equation}
Integrating both sides between $t_0$ and $t_i$,
\begin{equation*}
  \e^{- \beta_1 t_i} [Y(t_i) + \beta_0 / \beta_1] -
  \e^{- \beta_1 t_0} [Y(t_0) + \beta_0 / \beta_1] =
  \sigma_p \int_{t_0}^{t_i} \e^{- \beta_1 t} \dd W(t) \;.
\end{equation*}
Finally, substituting $Y(t_i) = y_i - \epsilon_i$ from eq.~\eqref{eq:data}, and solving for $y_i$,
\begin{equation} \label{eq:int0i}
  y_i = - \beta_0 / \beta_1 + \e^{\beta_1 (t_i -t_0)} \left(y_0 + \beta_0 / \beta_1
   - \epsilon_0\right) + \epsilon_i + \delta_{0i} \;;\quad i = 1, \ldots, n \;,
\end{equation}
where, from \eqref{eq:delta},
\begin{equation*}
  \delta_{0i} = \sigma_p \int_{t_0}^{t_i} \e^{\beta_1 (t_i -  t)} \dd W(t)
\end{equation*}
is a Gaussian random variable.
For a time-variant model, similar results are obtained replacing $\e^{- \beta_1 t}$ with the integrating factor $\exp \left[-\int \beta(t) \dd t \right]$, and moving $\sigma_p(t)$ inside the integral \citep[][Sec.~4.4]{kloeden92}. This can be generalized to the multidimensional case (Sec.~\ref{sec:multi}).

The means, variances and covariances determining the Gaussian probability density of a sample $(y_1, \ldots, y_n)$ can be obtained from eq.~\eqref{eq:int0i} \citep{garcia83}. The methods of Sec.~\ref{sec:tgm} can then be used, noting that $x_i$ and $y_i$ here correspond to $y_i$ and $z_i$ in Sec.~\ref{sec:tgm}. A sum of squares $\sum u_i^2$, calculated from the observations $x_1, \ldots, x_n$,  can be minimized  to compute ML or MAP estimates for the parameters $\vc \lambda, \beta_0, \beta_1, \sigma_p, \sigma_m$, and possibly $y_0$ and $\sigma_0$.

The covariance matrix of the $y_i$ is dense, all the $y_i$ are correlated because the integrals in $\delta_{0i}$ involve overlapping time intervals. It is more convenient to introduce an additional transformation from $\vc y$ to $\vc z$ that produces diagonal or tri-diagonal covariance matrices, simplifying computation. The idea is to consider changes between consecutive observations, integrating eq.~\eqref{eq:if} from $t_{i-1}$ to $t_i$ instead of from $t_0$ to $t_i$. The stochastic integrals are then independent, because the time intervals do not overlap. If there is an observation error, it affects the observed changes over the adjacent intervals, so that these are correlated, but intervals further apart remain independent. The relationship between consecutive observations is obtained substituting $i-1$ for 0 in eq.~\eqref{eq:int0i}:
\begin{equation} \label{eq:inti}
  y_i = - \beta_0 / \beta_1 + \e^{\beta_1 \Delta_i} \left(y_{i-1} + \beta_0 / \beta_1
   - \epsilon_{i-1}\right) + \epsilon_i + \delta_i \;;\quad i = 1, \ldots, n \;,
\end{equation}
where $\Delta_i \equiv t_i - t_{i-1}$ and
\begin{equation*}
  \delta_i = \sigma_p \int_{t_{i-1}}^{t_i} \e^{\beta_1 (t_i -  t)} \dd W(t) \;.
\end{equation*}
From \eqref{eq:delta}, the $\delta_i$ are independent Gaussian random variables with $\E[\delta_i] = 0$ and
\begin{equation*}
  \Var[\delta_i] = \sigma_p^2 \int_{t_{i-1}}^{t_i} \e^{2 \beta_1 (t_i - t)} \dd t
  =  \frac{\sigma_p^2}{2 \beta_1} \left(\e^{2 \beta_1 \Delta_i} - 1\right) \;.
\end{equation*}

Define
\begin{equation} \label{eq:zsde}
  z_i = y_i + \beta_0 / \beta_1 - \e^{\beta_1 \Delta_i} \left(y_{i-1} + \beta_0 / \beta_1\right)
  \;; \quad i = 1, \ldots, n \;.
\end{equation}
The Jacobian of the transformation $\vc y \to \vc z$ is 1, so that the $z_i$ can be directly substituted for the $y_i$ in the likelihood. From eq.~\eqref{eq:inti},
\begin{equation*}
  z_i = - \e^{\beta_1 \Delta_i} \epsilon_{i-1} + \epsilon_i + \delta_i \;.
\end{equation*}
Therefore,
\begin{subequations}\label{eq:zcorrs} \begin{align}
  E[z_i] &= 0 \;; \quad i = 1, \ldots, n \\
  \Var[z_1] &=
    \e^{2 \beta_1 \Delta_1} \sigma_0^2 + \sigma_m^2 +
     \frac{\sigma_p^2}{2 \beta_1} \left(\e^{2 \beta_1 \Delta_1} - 1\right) \\
  \Var[z_i] &=
    \e^{2 \beta_1 \Delta_i} \sigma_m^2 + \sigma_m^2 +
     \frac{\sigma_p^2}{2 \beta_1} \left(\e^{2 \beta_1 \Delta_i} - 1\right)
     \;;\quad i = 2, \ldots, n \\
  \Cov[z_i, z_{i-1}] &= \Cov[z_{i-1}, z_i] = - \e^{\beta_1 \Delta_i} \sigma_m^2
   \;; \quad i = 2, \ldots, n \\
   \Cov[z_i, z_j] &= 0 \quad\text{if}\quad \lvert i - j\rvert > 1
\end{align}\end{subequations}
\citep[][Sec.~7.5.3]{garcia83,seber-wild}.
It can be seen that the $z_i$ are the conditional residuals $y_i - E[y_i | y_{i-1}]$ and, as expected, they are uncorrelated if $\sigma_m = 0$, or are negatively correlated with the adjacent $z_{i-1}$ and $z_{i+1}$ otherwise.

If $\beta_1 = 0$, the relevant equations are:
\begin{gather*}
  \dd Y = \beta_0 \dd t + \sigma_p \dd W \\
  \dd (Y - \beta_0 t) = \sigma_p \dd W \\
  y_i - \epsilon_i - \beta_0 t_i - y_{i-1} + \epsilon_{i-1} + \beta_0 t_{i-1} = \delta_i\\
  z_i = y_i - y_{i-1} - \beta_0 \Delta_i = \epsilon_i - \epsilon_{i-1} + \delta_i \\
  \Var[z_1] = \sigma_m^2 + \sigma_0^2 + \sigma_p^2 \Delta_1 \\
  \Var[z_i] = 2 \sigma_m^2 + \sigma_p^2 \Delta_i \;;\quad i = 2, \ldots, n \\
  \Cov[z_i, z_{i-1}] = - \sigma_m^2 \;;\quad i = 2, \ldots, n \\
\end{gather*}

\subsection{Computational details}
\label{sec:comp}

To apply the methods of Section \ref{sec:tgm}, it is necessary to express the covariance matrix as the product of an unknown parameter $\sigma^2$ and a matrix $\mat C$ , see eq.~\eqref{eq:cov}. One could take $\sigma^2 = \sigma_m^2$ or $\sigma^2 = \sigma_p^2$, but this would fail if $\sigma^2 = 0$. A better option is
\begin{equation} \label{eq:sigma2}
  \sigma^2 = \sigma_m^2 + \sigma_p^2 \;.
\end{equation}
Then, with the re-parametrization
\begin{equation} \label{eq:eta}
  \sigma_m^2 = \sigma^2 \eta \;,\quad \sigma_p^2 = \sigma^2(1 - \eta) \;,\quad \sigma_0^2 = \sigma^2 \eta_0 \;,
\end{equation}
if $\beta_1 \neq 0$ the non-zero elements of $\mat C$ are
\begin{subequations}\label{eq:C} \begin{align}
  &c_{11} = \e^{2\beta_1 \Delta_i} \eta_0 + \eta +
    (1 - \eta) \frac{\e^{2 \beta_1 \Delta_1} - 1}{2 \beta_1} \\
  &c_{ii} = \e^{2\beta_1 \Delta_i} \eta + \eta +
    (1 - \eta) \frac{\e^{2 \beta_1 \Delta_i} - 1}{2 \beta_1}
     \;; \quad i = 2, \ldots, n \\
  &c_{i,i-1} = c_{i-1,i} = - \e^{\beta_1 \Delta_i} \eta \;; \quad i = 2, \ldots, n \;.
\end{align}\end{subequations}
If $\beta_1 = 0$,
\begin{subequations}\label{eq:C0} \begin{align}
  &c_{11} = \eta + \eta_0 + (1 - \eta) \Delta_1 \\
  &c_{ii} = 2 \eta + (1 - \eta) \Delta_i
     \;, \quad i = 2, \ldots, n \\
  &c_{i,i-1} = c_{i-1,i} = - \eta \;, \quad i = 2, \ldots, n \;.
\end{align}\end{subequations}
The ML method is invariant under re-parametrization \citep[][Sec.~5.1]{zacks}. In the optimization one needs to ensure that $0 \leq \eta \leq 1$.

It remains to find $\mat L$ such that $\mat C = \mat L \mat L'$, and then $\ln J = \sum \ln \lvert\pdr{\varphi_i}{x_i}\rvert - \sum \ln l_{ii}$, and $\vc u = \mat L^{-1} \vc z / J^{1/n}$ (see equations \eqref{eq:cholesky}--\eqref{eq:ucorr}). In \textit{R} this could be done using the \textit{base} functions \texttt{chol} and \texttt{baksolve}, or the corresponding sparse-matrix functions from package \textit{Matrix}. We show instead a generic method that makes efficient use of the structure of $\mat C$.

Because $\mat C$ is tri-diagonal, $\mat L$ has non-zero elements only on the diagonal and sub-diagonal, and $\mat C = \mat L \mat L'$ gives
\begin{align*} 
  &c_{11} = l_{11}^2 \\
  &c_{ii} = l_{i, i-1}^2 + l_{ii}^2 \;;\quad i=2, \ldots, n \\
  &c_{i, i-1} = l_{i, i-1} l_{i-1, i-1}  \;;\quad i=2, \ldots, n \;.
\end{align*}
Therefore, $\mat L$ can be calculated sequentially from
\begin{align*}
  &l_{11} = \sqrt{c_{11}} \\
  &\text{for } i=2, \ldots, n \\
  &\quad l_{i, i-1} = c_{i, i-1} / l_{i-1, i-1} \\
  &\quad l_{ii} = \sqrt{c_{ii} - l_{i, i-1}^2} \;.
\end{align*}
Similarly, from $\mat L^{-1} \vc z = \vc v$ and $\vc z = \mat L \vc v$,
\begin{align*}
  &z_1 = l_{11} v_1 \\
  &z_i = l_{i, i-1} v_{i-1} + l_{ii} v_{i}  \;;\quad i=2, \ldots, n \;,
\end{align*}
which can be solved as
\begin{align*}
  &v_1 = z_1 / l_{11} \\
  &\text{for } i=2, \ldots, n \\
  &\quad v_i = (z_i - l_{i, i-1} v_{i-1}) / l_{ii} \;.
\end{align*}
The two loops can be combined into one. The \textit{R} function \texttt{logdet.and.v} in \ref{app:sde} uses this to compute $\ln \lvert\mat L\rvert$ and $\vc v = \mat L^{-1} \vc z$.

Now the vector $\vc u$ can be computed with function \texttt{uvector} from \ref{app:sde}, given user-supplied functions for the transformation $x \to y$ and its derivative. The ML parameter estimates are found by minimizing the sum of squares $\sum u_i^2$. The same function \texttt{uvector} may be used to calculate estimates for $\sigma_p, \sigma_m$ and $\sigma_0$ and the maximized log-likelihood value.

If it is known that there are no observation errors ($\sigma_m = 0$), the $z_i$ are uncorrelated and the calculations can be greately simplified.

MAP estimates can be produced by adjoining a prior distribution as shown in Section \ref{sec:map}. A prior on $\eta$ such as a Beta distribution might be useful to force it away from the extremes $\eta = 0$ and $\eta = 1$.

\subsubsection{Example 2}
\label{sec:ex2}

\citet[][Ch.~6, Example 1]{pinheiro00} present data consisting of 6 height--age observations in each of 14 loblolly pine trees. This \texttt{Loblolly} dataset is included in the standard \textit{R} data collection \citep{R}. In this example we use only the first tree, tree \#301.
The height $h_i = H(t_i)$ is in feet, and the age or time $t$ is in years.

A suitable model is
\begin{equation*}
  \der{H^c}{t} = b(a^c - H^c) \;.
\end{equation*}
Integration gives the Bertalanffy-Richards growth curve
\[  H = a\left\{1 - \left[1 - (h_0 / a)^c\right] \exp\left[-b (t - t_0)\right]\right\}^{1 / c} \]
that, as explained above, for small $c$ approximates the Gompertz curve used by \citet{pinheiro00}, and for $c = -1$ coincides with the logistic from the original study. The parameter $a$ is the curve upper asymptote, $b$ is a time scale factor, and $c$ determines the location of the inflection point. Assume $H(0) = 0$.

Now introduce additive process noise,
\begin{equation} \label{eq:ex2sde}
  \dd H^c = b(a^c - H^c) \dd t + \sigma_p \dd W \;,
\end{equation}
and measurement error,
\begin{equation*}
  h_i^c = H(t_i)^c + \epsilon_i \;,\quad i = 1, \ldots, 6 \;,
\end{equation*}
with the $\epsilon_i$ independent normally distributed with mean 0 and variance $\sigma_m^2$.

The calculations are shown in detail in \ref{app:ex2}.
To facilitate convergence, initial estimates were first obtained fixing $\eta = 0.5$ ($\eta$ is the relative measurement variance from eqns.~\eqref{eq:sigma2}--\eqref{eq:eta}), minimizing the sum of squares of $\vc u$ using \texttt{uvector} and \texttt{nls}. Then, the sum of squares was minimized with $\eta$ constrained to be between 0 and 1. The result was $a = 72.55$, $b = 0.0967$, $c = 0.5024$, $\eta = 1.000$, $\sigma_p = 0$, $\sigma_m = 0.04865$, with a maximized log-likelihood value of -3.988.

The result is on the boundary $\eta = 1$, suggesting that most of the variability arises from measurement errors, with an indication of ill-conditioning and little effect of $\eta$ on the other parameter estimates. It has been found that with small datasets like this one it is often difficult to discriminate between environmental and observational sources of error \citep{garcia83}.

As an alternative to the additive noise, consider multiplicative process noise:
\begin{equation*}
    \dd H^c = b(a^c - H^c) (\dd t + \sigma_p \dd W)
     = b(a^c - H^c) \dd t + b \sigma_p (a^c - H^c) \dd W \;.
\end{equation*}
The Lamperti transform suggests
\begin{equation*}
  Y = \varphi(H, \vc\theta) = \ln\lvert a^c - H^c\rvert \;.
\end{equation*}
Under the Stratanovich interpretation of the SDE the ordinary calculus rules hold, and
\begin{equation*}
  \dd Y = - b \dd t + b \sigma_p \dd W
\end{equation*}
(ignoring an inconsequential change of sign in $W$). The measurement errors are assumed to be independent, of the form
\begin{equation*}
  y_i = Y(t_i) + \epsilon_i \;,\quad \epsilon_i \sim N(0, \sigma_m^2) \;.
\end{equation*}

There seem to be good reasons to prefer Stratanovich to Ito in this instance \citep[][Section 6.1]{kloeden92}. For the more common Ito interpretation, it is found from Ito's formula \citep[][eq.~(4.3.14)]{gardiner94} that the equations differ only by a change of parameters: $\dd Y = - (b + \tfrac{1}{2} b^2 \sigma_p^2) \dd t + b \sigma_p \dd W$. We keep the previous formulation.

The ML estimation for the multiplicative noise version gave
$a = 77.11$, $b = 0.08405$, $c = 0.54946$, $\eta = 1.000$, $\sigma_p = 0$, $\sigma_m = 0.01577$, with a maximized log-likelihood value of -3.568 (details in \ref{app:ex2}).
The log-likelihood is not significantly different from the one for the additive model, and on a graph the two curves are quite close (Figure \ref{fig:figure}).

\fig{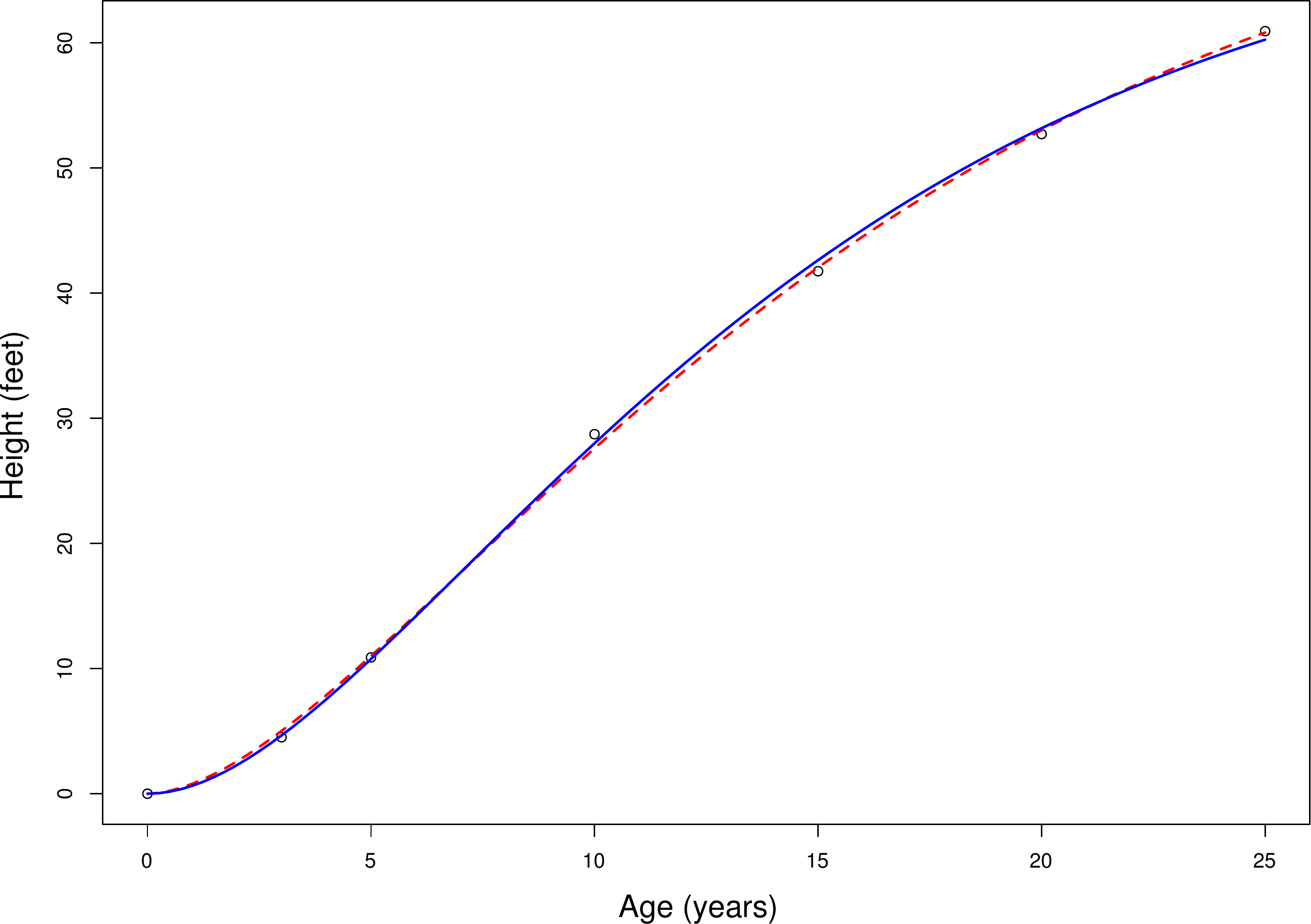}{Data and fitted models for Example 2. Continuous curve: additive process noise. Dashed curve: multiplicative process noise.}

Multiplicative models for most classical growth curves can be derived using Box-Cox transformations, as special cases of either $Y = \ln[-(H / a)^{(c_1)}]^{(c_2)}$ or $Y = \ln[-(1 - H/ a)^{(c_1)}]^{(c_2)}$ \citep{grex}. The Bertalanffy-Richards family corresponds to $c_2 = 0$

\subsection{Hierarchical models}
\label{sec:hier}

Often the data consists of several measurements on each of a number of individuals or sampling or experimental units (\emph{units}, for short). For instance, the height measurements at various ages on each of 14 trees in the \texttt{Loblolly} dataset of Example 2. This is known as panel, repeated measures, or longitudinal data, and gives rise to hierarchical or multilevel models; the case of two hierarchical levels is discussed here. Some parameters may vary among units (\emph{local}), while others are common to all units (\emph{global}), possibly after a re-parametrization of the original model. Local parameters may be treated simply as fixed unknown values, called nuisance parameters if they are not the main object of interest \citep{garcia83}. More commonly, in mixed-effects models, the locals are viewed as ``random'', possibly reflecting sampling from some hypothetical super-population \citep{pinheiro00,snijders03}. Mixed-effects SDE examples include \citet{donnet08}, \citet{klim09}, \citet{picchini10}, \citet{driver17}, and \citet{whitaker17}.

In the derivations of Section \ref{sec:tgm}, let us substitute a double subscript $ij$ for $i$, indicating observation $i = 1, \ldots, n_j$ in unit $j = 1, \ldots, m$. Then, the observations become $\vc y = (y_{11}, \ldots, y_{n_m, m})$, or $\vc y = (\vc y_1, \ldots, \vc y_m)$, where $\vc y_j$ is the vector of $n_j$ observations $y_{ij}$ in unit $j$. With these notational changes, it is seen that the arguments and results of Section \ref{sec:tgm} are still valid, provided that $\sigma^2$ is common to all units (global). However, some computations can be simplified:

It is commonly assumed that the units are statistically independent. Assume also that the transformations involve only observations in the same unit, that is, $z_{ij} = \varphi_{ij}(\vc y_j, \vc\theta_j)$. The parameters $\vc\theta_j$ can vary from unit to unit. Then, assuming that $\sigma^2$ is common to all units, eq.~\eqref{eq:ucorr} becomes
\begin{equation} \label{eq:uij}
  \vc u = \frac{\left(\mat L_1^{-1} \vc\epsilon_1(\vc y_1, \vc\theta_1), \ldots,
    \mat L_m^{-1} \vc\epsilon_m(\vc y_m, \vc\theta_m)\right)}{J^{1/n}}
\end{equation}
where $J$ in eq.~\eqref{eq:Jcorr} can be obtained from
\begin{equation*}
  J = \abs\left(\left\lvert\pder{\mat L^{-1} \vc \varphi}{\vc y}\right\rvert\right)
    = \abs\left(\prod_{j=1}^m \left\lvert\pder{\mat L_j^{-1} \vc\varphi_j}{\vc y_j}\right\rvert\right) \;.
\end{equation*}

Applying this to the SDEs, one can compute $\vc v_j = \mat L_j^{-1} \vc z_j$ and $\ln J_j$for each unit as before (Section \ref{sec:comp}), and obtain $\vc u = \vc v / J^{1/n}$ with $\vc v = (\vc v_1,\ldots, \vc v_m)$ and $\ln J =  \sum \ln J_j$.

To ensure a global $\sigma^2$ under realistic circumstances, assume that $\sigma_m$ can be written as the product of a possibly local multiplier $\mu_m(\vc\theta_j) = \mu_{mj}$ and a global parameter $\sigma_M$, and similarly for $\sigma_0$ and $\sigma_p$:
\begin{equation} \label{eq:mu}
  \sigma_m = \mu_{mj} \sigma_M \;,\quad  \sigma_0 = \mu_{0j} \sigma_Z \;,\quad
    \sigma_p = \mu_{pj} \sigma_P
   \;.
\end{equation}
Then, define
\begin{equation} \label{eq:sigmaj}
  \sigma^2 = \sigma_M^2 + \sigma_P^2 \;,\quad \eta = \sigma_M^2 / \sigma^2
   \;,\quad \eta_0 = \sigma_Z^2 / \sigma^2 \;,
\end{equation}
so that in eqs.~\eqref{eq:zcorrs} we substitute
\begin{equation*}
  \sigma_m^2 = \sigma^2 \mu_{mj}^2 \eta \;,\quad \sigma_0^2 =
   \sigma^2 \mu_{0j}^2 \eta_0 \;,\quad \sigma_p^2 = \sigma^2 \mu_{pj}^2 (1 - \eta)
\end{equation*}
to obtain the analogous of eq.~\eqref{eq:C} for unit $j$
\begin{subequations}\label{eq:Cj} \begin{align}
  &c_{11j} = \e^{2\beta_{1j} \Delta_{1j}} \mu_{0j}^2 \eta_0 + \mu_{mj}^2 \eta +
    (1 - \eta) \mu_{pj}^2 \frac{\e^{2 \beta_{1j} \Delta_{1j}} - 1}{2 \beta_{1j}} \\
  &c_{iij} = \e^{2\beta_{1j} \Delta_{ij}} \mu_{mj}^2 \eta + \mu_{mj}^2 \eta +
    (1 - \eta) \mu_{pj}^2 \frac{\e^{2 \beta_{1j} \Delta_{ij}} - 1}{2 \beta_{1j}}
     \;, \quad i = 2, \ldots, n \\
  &c_{i,i-1,j} = c_{i-1,i,j} = - \e^{\beta_{1j} \Delta_{ij}} \mu_{mj}^2 \eta \;,
  \quad i = 2, \ldots, n
\end{align}\end{subequations}
with $\Delta_{ij} = t_{ij} - t_{i-1, j}$.

In the case $\beta_{1j} = 0$, eqs.~\eqref{eq:C0} become
\begin{subequations}\label{eq:C0j} \begin{align}
  &c_{11j} = \mu_{mj}^2 \eta + \mu_{0j}^2 \eta_0 + (1 - \eta) \mu_{pj}^2 \Delta_{1j} \\
  &c_{iij} = 2 \mu_{mj}^2 \eta + (1 - \eta) \mu_{pj}^2 \Delta_{ij}
     \;, \quad i = 2, \ldots, n \\
  &c_{i,i-1,j} = c_{i-1,i,j} = - \mu_{mj}^2 \eta \;, \quad i = 2, \ldots, n \;.
\end{align}\end{subequations}

An extension of the function \texttt{uvector} to handle hierarchical models is listed in \ref{app:hier}.

\subsubsection{Fixed local parameters}
\label{sec:fixed}

Viewing both  global and local parameters as unknown constants, estimation is relatively simple in \textit{R}, where function \texttt{nls} allows vector-valued parameters. This feature seems to be documented only in the last example of the \texttt{nls} help page, which appears also on page 219 of \citet{venables02}.

The method performs reasonably well for at least several hundred parameters. Presumably, \texttt{nls} exploits special structure through matrix partitioning strategies, as in \citet{griewank82}, \citet{garcia83}, or \citet{soo92}. It was verified that the approach gives exactly the same results as the method of \citet{garcia83}, which is based on direct optimization of the log-likelihood (software and documentation available at \url{http://forestgrowth.unbc.ca/sde}). Computation times are much longer than those of Garc\'{\i}a's compiled Fortran code, but that is not an impediment with modern computing hardware.

\paragraph{Example 3.}
\label{sec:ex3}
Consider fitting Richards SDE models simultaneously to all the 14 trees of the \texttt{Loblolly} dataset of Example 2. The parametrization can be important here, so re-write eq.~\eqref{eq:ex2sde} using the Box-Cox transformation of eq.~\eqref{eq:bc}, ensuring that $a$ and $b$ are proper scale parameters:
\begin{align*}
  Y &= \left(H / a\right)^{(c)} \\
  \dd Y &= - Y \dd (b t) + \sigma_P \dd W(b t) = - b Y \dd t +
    \sqrt{\lvert b\rvert}\sigma_P \dd W(t) \;.
\end{align*}
Assume that the measurement error is negligible compared to the process noise.

See \ref{app:ex3} for details of the calculations.

First, take $a$ as local, i.~e., the asymptotes $a_j$ vary from tree to tree. The ML estimates are found to be:
\[  
    a_1 = 68.37, \dots, a_{14} = 78.84, b = 0.09472, c = 0.4918, \sigma_P = 0.03359 \;,
\]
with a log-likelihood of -88.4. For comparing models with different numbers of parameters one can use Akaike's Information Criterion AIC = 210.8, or Schwarz's Bayesian Information Criterion BIC = 252.1. These indices penalize twice the negative log-likelihood by subtracting a quantity that increases with the number of estimated parameters.

Trying with $a$ global and $b$ local gives
\[
    a = 73.08, b_1 = 0.08912, \dots, b_{14} = 0.1031, c = 0.4916, \sigma_P = 0.03231 \;,
\]
with log-likelihood = -85.2, AIC = 204.3, BIC = 245.6. 
The log-likelihood (or equivalently, the AIC or BIC) indicates that this model fits the data slightly better than the one with $a$ local.

Taking both $a$ and $b$ as locals gave worse AIC and BIC values than those for the one-local versions. Other structures could be defined by re-parametrization, substituting functions of other global and local parameters for $a$, $b$ and $c$ \citep{garcia83}.

\subsubsection{Random local parameters}
\label{sec:random}

A more popular alternative is to view local parameters as, in some sense, random. This is known as a \emph{nonlinrar mixed-effects model} \citep{pinheiro00}. The easiest situation to understand and justify is where the units are thought to be a random sample from a population where the parameters follow certain frequency function. We adopt the usual assumption of Gaussian local parameter distributions. In mixed-effects modelling terminology, the units are sometimes called \emph{groups} (of observations), the global parameters and the means of the locals are \emph{fixed effects}, and the deviations of the local parameters from their means are \emph{random effects}.

With this approach the number of unknown parameters is much reduced, instead of one local parameter value for each unit it is only necessary to estimate a mean and a variance (and possibly covariances among locals). On the other hand, there are additional model assumptions, and the estimation algorithms are more complicated. The model may be questionable when, as is often the case, rather than being a simple random sample the units are chosen to have a more efficient coverage of a range of conditions.

A mixed-effects SDE model can be estimated by maximum likelihood using standard packages and the $u_{ij}$ computations above.
In \textit{R} one can use function \texttt{nlme} from the package of the same name, or \texttt{nlmer} from package \textit{lme4}. These functions do not support bounds on parameter values, so that in SDEs with observation errors it would be necessary to do multiple runs with given values of $\eta$, or alternatively, to embed the mixed-effects estimation within a one-dimensional optimization over $\eta$ (e.~g., using \texttt{optimize}).

\paragraph{Example 4.}
\label{sec:ex4}
The $b$-local version from Example 3 was fitted as a mixed-effects model using the \texttt{nlme} \textit{R} package (\ref{app:ex4}). It was necessary to relax the default tolerance in order to achieve convergence. The ML estimates for the global parameters were $a = 73.43$ and $c = 0.4938$. In this formulation $b$ is a random variable, so that it does not make sense to speak of ``estimates'', but the software provides ``predictions'' $b_1 = 0.08973, \dots, b_{14} = 0.09964$, and a mean $0.09381$.

The log-likelihood was $-101.78$, which is not comparable to that in Example 3 because there are less free parameters in the mixed-effects approach. The AIC and BIC should be more informative. 
The AIC = 213.6 compared to the AIC = 204.3 in Example 3 suggests that the fixed-parameters model is better than the mixed-effects version. However, the opposite is true for the BIC that penalizes the number of parameters more heavily, $225.7$ \emph{vs.} $245.6$.

\subsection{Multivariate SDEs}
\label{sec:multi}

Some applications use a system of SDEs or, equivalently, a vector-valued SDE \citep[e.~g.,][]{gdns,klim09,picchini11,donnet13,driver17}.
Analogously to the single-variable case, we assume that there is a random $p$-dimensional  vector $\vc Y$ that follows a linear SDE
\begin{equation} \label{eq:veclin}
  \dd \vc Y = (\mat A \vc Y + \vc b) \dd t + \mat S \dd \vc W \;.
\end{equation}
Here $\mat A$ and $\mat S$ are $p \times p$ matrices, and $\vc b$ is a $p$-vector, any or all of them dependent on parameters in $\vc\theta$, and $\vc W$ is a vector of $p$ independent Wiener processes.

The observations $\vc x_i$ at time $t_i$ are related to $\vc Y$ by a one-to-one transformation
\begin{equation} \label{eq:vecphi}
  \vc y_i = \vc\varphi(\vc x_i, \vc\theta) \;,
\end{equation}
and may be subject to observation error according to
\begin{equation} \label{eq:vecobs}
  \vc y_i = \vc Y(t_i) + \vc\epsilon_i \;,\quad \vc\epsilon_i \sim N(\vc 0, V) \;.
\end{equation}
The errors $\vc\epsilon_i$ are uncorrelated across observation times, and independent of the process noise $\vc W$. The initial conditions are $Y(t_0) = \vc y_0 + \vc\epsilon_0$, with $\vc\epsilon_0 \sim N(\vc 0, \mat V_0)$. All this can apply to units with different parameter values $\vc\theta_j$ and expected initial conditions $\vc y_{0j}$; non-essential indices will be omitted.

A useful example of transformation $\vc y = \vc\varphi(\vc x, \vc\theta)$ is
\begin{equation*}
  y_k = x_1^{c_{k1}} x_2^{c_{k2}} \cdots x_p^{c_{kp}} \;,\quad k = 1, \ldots, p \;,
\end{equation*}
which can be conveniently denoted as $\vc y = \vc x^{\mat C}$, defining $\vc x^{\mat C} \equiv \exp(\mat C \ln \vc x)$. This gives a flexible multivariate analog of the Bertalanffy-Richards SDE \citep{gdns}.

In the time-invariant case the matrix exponential $\e^{-\mat A t}$ is an integrating factor \citep[see][or the \textit{expm} package in \textit{R} for definitions and computation of the matrix exponential]{moler03}. For non-singular $\mat A$,
\begin{equation*} 
  \dd \left[\e^{-\mat A t} \left(\vc Y + \mat A^{-1}\vc b\right)\right] = \e^{-\mat A t} \mat S \dd\vc W
\end{equation*}
(c.~f.\ eq.~\eqref{eq:if}). Then,
\begin{equation*}
  \e^{-\mat A t_i} \left[\vc Y(t_i) + \mat A^{-1}\vc b\right] - \e^{-\mat A t_{i-1}} \left[\vc Y(t_{i-1}) + \mat A^{-1}\vc b\right] = \int_{t_{i-1}}^{t_i} \e^{-\mat A t} \mat S \dd\vc W(t) \;.
\end{equation*}
Multiplying by $\e^{\mat A t_i}$ and substituting $\vc Y(t_i) = \vc y_i - \vc\epsilon_i$,
\begin{equation} \label{eq:vecz}
  \vc z_i \equiv \vc y_i + \mat A^{-1}\vc b - \e^{\mat A \Delta_i} \left(\vc y_{i-1} + \mat A^{-1}\vc b\right) = \vc\epsilon_i - \e^{\mat A \Delta_i} \vc\epsilon_{i-1} + \vc\delta_i
\end{equation}
with $t_i - t_{i-1} \equiv \Delta_i$ and
\begin{gather*}
  \vc\delta_i = \int_{t_{i-1}}^{t_i} \e^{\mat A (t_i - t)} \mat S \dd\vc W(t) \;, \\
  \E[\vc\delta_i] = 0 \;,
\end{gather*}
\begin{equation} \label{eq:Vdelta}
    \Var[\vc\delta_i] = \int_{t_{i-1}}^{t_i} \e^{\mat A (t_i - t)} \mat S \mat S'
       {\e^{\mat A (t_i - t)}}' \dd t
       = \int_0^{\Delta_i} \e^{\mat A t} \mat S \mat S' {\e^{\mat A t}}' \dd t
\end{equation}
similarly to eq.~\eqref{eq:delta} \citep[][Sec.~4.4.6, and Sec.~4.4.9 for the time-variant case]{gardiner94}.

Therefore, eq.~\eqref{eq:vecz} defines a transformation $\{\vc y_i\} \to \{\vc z_i\}$ with unit Jacobian, such that the $\vc z_i$ (or $\vc z_{ij}$) are Gaussian with
\begin{subequations}\begin{align}
  \E[\vc z_i] &= \vc 0 \;;\quad i = 1, \ldots, n \\
  \Var[\vc z_1] &= \mat V + \e^{\mat A \Delta_0} \mat V_0 {\e^{\mat A \Delta_0}}' + \Var[\vc\delta_1] \\
  \Var[\vc z_i] &= \mat V + \e^{\mat A \Delta_i} \mat V {\e^{\mat A \Delta_i}}' + \Var[\vc\delta_i]
    \;;\quad i = 2, \ldots, n \\
  \Cov[\vc z_i, \vc z_{i-1}] &= - \e^{\mat A \Delta_i} \mat V \;;\quad i = 2, \ldots, n \;.
\end{align}\end{subequations}

\citet{gdns} obtained explicit likelihoods for systems where oscillations are not admissible, for which $\mat A = \mat P^{-1} \mat\Lambda \mat P$, $\e^{\mat A \Delta_i} = \mat P^{-1} \e^{\mat\Lambda \Delta_i} \mat P$, and $\mat\Lambda$ is a diagonal matrix of real eigenvalues. That includes dynamical systems with monotonic state variables, for instance, forest plantations where diameter and height increases and number of trees decreases over time. There it is convenient to take the elements of $\mat \Lambda$ and $\mat P$ instead of those of $\mat A$ as parameters to be estimated.

For general $\mat A$, it is found through integration by parts of eq.~\eqref{eq:Vdelta} that:
\begin{proposition}
 $\Var[\vc\delta_i]$ satisfies
\begin{equation}\label{eq:lyap}
  \mat A \Var[\vc\delta_i] + \Var[\vc\delta_i] \mat A' =
   \e^{\mat A \Delta_i} \mat S \mat S' {\e^{\mat A \Delta_i}}' - \mat S \mat S'\;.
\end{equation}
\end{proposition}
\begin{proof} Making use of the fact that $\mat A$ and $\e^{\mat A t}$  commute,
\begin{align*}
    \dd\, &\e^{\mat A t} \mat S {\mat S}' {\e^{\mat A t}}'
    = \dd\, \left(\e^{\mat A t} \mat S \right) \left(\e^{\mat A t} \mat S \right)' \\
    &= \left(\mat A \e^{\mat A t} \mat S \dd t \right) \left(\e^{\mat A t} \mat S \right)'
    + \left(\e^{\mat A t} \mat S \right) \left(\mat A \e^{\mat A t} \mat S \dd t \right)' \\
    &= \mat A \left(\e^{\mat A t} \mat S {\mat S}' {\e^{\mat A t}}' \dd t \right)
    + \left(\e^{\mat A t} \mat S {\mat S}' {\e^{\mat A t}}' \dd t \right) {\mat A}' \;.
\end{align*}
Integrating between 0 and $\Delta_i$ and using eq.~\eqref{eq:Vdelta},
\[
   \e^{\mat A \Delta_i} \mat S \mat S' {\e^{\mat A \Delta_i}}' - \mat S \mat S'
   = \mat A \Var[\vc\delta_i] + \Var[\vc\delta_i] \mat A' \;.
\]
\end{proof}
Eq.~\eqref{eq:lyap} is known as a \emph{continuous-time Lyapunov equation}. A recent review of solution methods is contained in \citet[][Sec.~5]{simoncini16}.

The likelihood for the linear SDE can also be derived less directly by writing it as a product of the conditional densities of $\vc y_i$ given the previous $\vc y_0, \ldots, \vc y_{i-1}$. The conditional expectations and variance-covariance matrices  can be computed recursively using the Kalman filter \citep[e.~g.,][Sec.~4.2.1]{donnet13}.

Regardless of how the likelihood is produced, the Furnival-Box-Cox device following the methods of Sec.~\ref{sec:sdes} can then be used to reduce the likelihood maximization to minimizing a sum of squares. Multidimensionality does not introduce additional conceptual issues, although the details and notation can be complicated. MAP estimation as in Section \ref{sec:map} is also feasible.

Taking advantage of the ML re-parametrization invariance, it is preferable to estimate Cholesky factors instead of optimizing directly over variance-covariance matrices like $\mat V$ or $\mat V_0$ \citep[e.~g.,][]{gdns}. This ensures that the matrix will be positive-definite, as it should. The case of $\mat S$ in eq.~\eqref{eq:veclin} is more delicate, because the likelihood function depends only on the product $\mat S \mat S'$, or on the Cholesky factor of $\mat S \mat S'$. This means that a full square matrix $\mat S$ in unidentifiable, that is, different $\mat S$'s produce the same observation distribution. Therefore, a triangular matrix should replace $\mat S$ in the optimization. \citet{driver17} directly assume that $\mat S$ in eq.~\eqref{eq:veclin} is triangular; it is not always clear how the issue is handled in other software.

\section{Conclusions}
\label{sec:concl}

\citet{furnival61} and \citet{box-cox} showed how maximum likelihood estimation involving nonlinear transformations of Gaussian dependent variables can be achieved by minimizing a sum of squares. The technique has many potential applications, the focus here is on stochastic differential equation modelling. A similar idea can be applied to Bayesian maximum a posteriori estimates.

SDEs are more realistic than traditional approaches in longitudinal data analysis and other applications, but wider acceptance has been hampered by mathematical and computational complexity. Linear SDEs allow explicit solutions, and are suitable for compartmental and other models \citep[][Sec.~8.8]{cleur00,kristensen05,klim09,favetto10,cuenod11,driver17,seber-wild}. However, they may not be appropriate in more general situations. Much more flexibility is achieved by allowing nonlinear transformations of variables in linear SDEs. The resulting so-called reducible SDEs are mathematically tractable, and extensions of the Furnival-Box-Cox method make parameter estimation possible using standard nonlinear least-squares or mixed-effects software.

Mixed-effects modelling has become a \emph{de facto} requirement for publication in many scientific journals, but the underlying assumptions are not always justifiable. Often, a model with group fixed-effects may be preferable. A somewhat obscure feature of the \texttt{nls} function in \textit{R} makes the implementation of such fixed parameters models particularly simple.

The proposed methods have been discussed and demonstrated in detail for univariate SDEs. Results match those computed by direct maximization of the likelihood functions. Extensions to the multivariate case are outlined more briefly. There are no conceptual difficulties in these, although the best choice of implementation specifics is not obvious, and likely to be problem-dependent.

Being able to use well tested and familiar standard software should make SDE modelling more transparent and accessible. Within the class of reducible SDEs, the methodology can also provide more flexibility in model specification compared to special-purpose software packages. Ease of use and understanding may contribute to a wider adoption of these models.

\section*{Acknowledgements}
\label{sec:ack}
\addcontentsline{toc}{section}{Acknowledgements}

I am grateful to an anonymous reviewer and to Dr.~Lance Broad
for thoughtful and detailed comments that contributed to improve the text.

\appendix

\section*{Appendices --- \textit{R} computer code}
\label{app:R}


\section{Example 1, Section \ref{sec:sol1}}
\label{app:sol1}

{\small\begin{verbatim}
> bc <- function(y, lambda){ # Box-Cox transformation
+     if (abs(lambda) < 1e-9) log(y)
+     else (y^lambda - 1) / lambda
+ }
> x <- MASS::GAGurine$Age # data
> y <- MASS::GAGurine$GAG
> lm(log(y) ~ I(x - 1)) # find suitable starting points
> # (with lambda.y = 0, lambda.x = 1)

Call:
lm(formula = log(y) ~ I(x - 1))

Coefficients:
(Intercept)     I(x - 1)
     2.8522      -0.1139

> # Minimize sum of squares:
> gm <- exp(mean(log(y)))  # geometric mean
> nls(rep(0, length(y)) ~ (bc(y, lambda.y) - beta0 - beta1 *
+ bc(x, lambda.x)) / gm^(lambda.y - 1), start = c(beta0=2.9,
+ beta1=-0.11, lambda.x=1, lambda.y=0))
Nonlinear regression model
  model: rep(0, length(y)) ~ (bc(y, lambda.y) - beta0 - beta1 * bc(x,
       lambda.x))/gm^(lambda.y - 1)
   data: parent.frame()
   beta0    beta1 lambda.x lambda.y
  3.3142  -0.3502   0.4249   0.1032
 residual sum-of-squares: 3214

Number of iterations to convergence: 10
Achieved convergence tolerance: 1.808e-06
> # Estimate for sigma:
> n <- length(y)
> gm^(0.1032 - 1) * sqrt(3214 / n)
[1] 0.383853
>
> # Round the lambdas for a more parsimonious model:
> lm(log(y) ~ sqrt(x))

Call:
lm(formula = log(y) ~ sqrt(x))

Coefficients:
(Intercept)      sqrt(x)
     3.3665      -0.5069
\end{verbatim}}

The results are verified by direct minimization of the negative log-likelihood from eq.~\eqref{eq:Lraw},
\[
  -\ln L = \tfrac{n}{2} \ln(2 \pi \sigma^2) +
    \frac{\sum \epsilon_i^2}{2 \sigma^2} - \ln J \;:
\]
{\small\begin{verbatim}
> neglogL <- function(theta){# theta = (beta0,beta1,lambda.x,lambda.y,sigma)
+ (n/2) * log(2 * pi * theta[5]^2) + sum((bc(y, theta[4]) - theta[1] -
+ theta[2] * bc(x, theta[3]))^2) / (2*theta[5]^2) - n * (theta[4] - 1) *
+ log(gm)}
> optim(par = c(3, -0.4, 0.4, 0.1, 0.4), fn = neglogL)
$par
[1]  3.3140471 -0.3502047  0.4248816  0.1031103  0.3837414

$value
[1] 810.6863

$counts
function gradient
     422       NA

$convergence
[1] 0

$message
NULL
\end{verbatim}}

\section{Functions for ML parameter estimation in SDE models, Section \ref{sec:comp}}
\label{app:sde}

The following function computes $\ln \lvert\mat L\rvert$ and $\vc v = \mat L^{-1} \vc z$.
{\small\begin{verbatim}
logdet.and.v <- function(
# Calculate log(det(L)) and v = L^-1 z
  cdiag,  # vector with the diagonal elements of C, c[i, i]
  csub,   # vector with sub-diagonal c[i, i-1] for i > 1
  z)
{
  v <- z
  ldiag <- sqrt(cdiag[1])
  logdet <- log(ldiag)
  v[1] <- z[1] / ldiag
  for (i in 2:length(z)) {
    lsub <- csub[i] / ldiag
    ldiag <- sqrt(cdiag[i] - lsub^2)
    logdet <- logdet + log(ldiag)
    v[i] <- (z[i] - lsub * v[i - 1]) / ldiag
  }
  list(logdet = logdet, v = v)
}
\end{verbatim}}

The function \texttt{uvector} that follows computes the vector $\vc u$, given \texttt{logdet.and.v} and user-supplied functions \texttt{phi(x, theta)} for the transformation $x \to y$ and \texttt{phiprime(x, theta)} for its derivative.
Once $\sum u_i$ is minimized, \texttt{uvector} can be used to calculate the ML estimates for $\sigma_p, \sigma_m$ and $\sigma_0$, and the maximized log-likelihood value, if \texttt{final = TRUE} and the optimal parameter values are passed in the arguments.
{\small\begin{verbatim}
uvector <- function(
  # If final = FALSE (default), returns vector whose sum of squares
  #   should be minimized.
  # If final = TRUE, passing the ML parameter estimates returns a
  #   list with the sigma estimates and the maximized log-likelihood.
  # Requires a transformation function y = phi(x, theta), and a
  #   function phiprime(x, theta) for the derivative dy/dx, where
  #   theta is a list containing all the parameters.
  x, t,  # data vectors
  beta0, beta1, eta, eta0, x0, t0,  # SDE parameters
  lambda,  # named list with additional parameter(s) for phi
  final = FALSE)
{
  theta <- c(list(beta0=beta0, beta1=beta1, eta=eta, eta0=eta0,
                  x0=x0, t0=t0), lambda)
  s <- order(t); t <- t[s]; x <- x[s] # ensure increasing t
  n <- length(x)
  y <- phi(x, theta)
  y0 <- phi(x0, theta)
  Dt <- diff(c(t0, t))
  if (beta1 != 0) {
    ex <- exp(beta1 * Dt)
    ex2 <- ex^2
    z <- y + beta0 / beta1 - ex * (c(y0, y[-n]) + beta0 / beta1)
    cdiag <- ex2 * eta + eta + (1 - eta) * (ex2 - 1) / (2 * beta1)
    cdiag[1] <- cdiag[1] - ex2[1] * (eta - eta0)
    csub <- -ex * eta
  } else { # beta1 == 0
    z <- y - c(y0, y[-n]) - beta0 * Dt
    cdiag <- 2 * eta + (1 - eta) * Dt
    cdiag[1] <- cdiag[1] - eta + eta0
    csub <- -rep(eta, n)
  }
  ld.v <- logdet.and.v(cdiag, csub, z)
  logJ <- sum(log(abs(phiprime(x, theta)))) - ld.v$logdet
  Jn <- exp(logJ / n)   # J^(1/n)
  u <- ld.v$v / Jn
  if (!final) return (u)  # "normal" exit
  # Else, at optimum, calculate sigma.p, sigma.m and sigma.0
  #   estimates, and the log-likelihood:
  ms <- sum(u^2) / n  # mean square
  sigma2 <- Jn^2 * ms  # estimate for sigma^2
  list(sigma.p = sqrt((1 - eta) * sigma2), sigma.m = sqrt(eta *
       sigma2), sigma.0 = sqrt(eta0 * sigma2), loglikelihood =
       -(n / 2) * (log(ms) + log(2 * pi) + 1))
}
\end{verbatim}}

\section{Example 2, Section \ref{sec:ex2}}
\label{app:ex2}

Extract tree 301:
{\small\begin{verbatim}
> (lob301 <- Loblolly[Loblolly$Seed == 301, ])
   height age Seed
1    4.51   3  301
15  10.89   5  301
29  28.72  10  301
43  41.74  15  301
57  52.70  20  301
71  60.92  25  301
\end{verbatim}}

\subsection{Additive process noise}
\label{app:add}

The functions for the transformation $H^c$ and its derivative $c H^{c-1}$ are
{\small\begin{verbatim}
> phi <- function(H, theta) H^theta$c
> phiprime <- function(H, theta) with(theta, c * H^(c - 1))
\end{verbatim}}
To facilitate convergence, start by fixing $\eta = 0.5$ ($\eta$ is the relative measurement variance from eqns.~\eqref{eq:sigma2}--\eqref{eq:eta}).
Minimizing the sum of squares of $\vc u$ using \texttt{nls}, with an initial guess $a = 70$, $b = 0.1$, $c = 1$,
{\small\begin{verbatim}
> nls(~ uvector(x = height, t = age, beta0 = b * a^c, beta1 = -b, eta = 0.5,
+  eta0 = 0, x0 = 0, t0 = 0, lambda = list(c = c)), data = lob301,
+  start = list(a = 70, b = 0.1, c = 1))
Nonlinear regression model
  model: 0 ~ uvector(x = height, t = age, beta0 = b * a^c, beta1 = -b,
         eta = 0.5, eta0 = 0, x0 = 0, t0 = 0, lambda = list(c = c))
   data: lob301
       a        b        c
71.96058  0.09947  0.49217
 residual sum-of-squares: 1.829

Number of iterations to convergence: 6
Achieved convergence tolerance: 1.145e-06
\end{verbatim}}

The option \texttt{algorithm = "port"} allows upper and lower bounds on the parameters, and we use it now for $0 \leq \eta \leq 1$:
{\small\begin{verbatim}
> nls(~ uvector(x = height, t = age, beta0 = b * a^c, beta1 = -b, eta = eta,
+ eta0 = 0, x0 = 0, t0 = 0, lambda = list(c = c)), data = lob301,
+ start = list(a = 70, b = 0.1, c = 0.5, eta = 0.5), algorithm = "port",
+ lower = c(0, 0, 0, 0), upper = c(100, 1, 2, 1))
Nonlinear regression model
  model: 0 ~ uvector(x = height, t = age, beta0 = b * a^c, beta1 = -b,
         eta = eta, eta0 = 0, x0 = 0, t0 = 0, lambda = list(c = c))
   data: lob301
      a       b       c     eta
72.5459  0.0967  0.5024  1.0000
 residual sum-of-squares: 1.327

Algorithm "port", convergence message: relative convergence (4)
\end{verbatim}}

Finally, calculate the variance estimates and the maximized log-likelihood:
{\small\begin{verbatim}
> uvector(x = lob301$height, t = lob301$age, beta0 = 0.0967 * 72.5459^0.5024,
+ beta1 = -0.0967, eta = 1, eta0 = 0, x0 = 0, t0 = 0,
+ lambda = list(c = 0.5024), final = TRUE)
$sigma.p
[1] 0

$sigma.m
[1] 0.04865072

$sigma.0
[1] 0

$logLikelihood
[1] -3.9882
\end{verbatim}}

\subsection{Multiplicative process noise}
\label{app:mult}

{\small\begin{verbatim}
> phi <- function(H, theta)
+   with(theta, log(abs(a^c - H^c)))
> phiprime <- function(H, theta)
+   with(theta, - c * H^(c - 1) / (a^c - H^c))
> nls(~ uvector(x = height, t = age, beta0 = -b, beta1 = 0, eta = eta,
+ eta0 = 0, x0 = 0, t0 = 0, lambda = list(a = a, c = c)), data = lob301,
+ start = list(a = 72, b = 0.1, c = 0.5, eta = 0.5), algorithm = "port",
+ lower = c(0, 0, 0, 0), upper = c(100, 1, 2, 1))
Nonlinear regression model
  model: 0 ~ uvector(x = height, t = age, beta0 = -b, beta1 = 0, eta = eta,
         eta0 = 0, x0 = 0, t0 = 0, lambda = list(a = a, c = c))
   data: lob301
       a        b        c      eta
77.10687  0.08405  0.54946  1.00000
 residual sum-of-squares: 1.154

Algorithm "port", convergence message: relative convergence (4)
> uvector(x = lob301$height, t = lob301$age, beta0 = -0.08405, beta1 = 0,
+ eta = 1, eta0 = 0, x0 = 0, t0 = 0,
+ lambda = list(a = 77.10687, c = 0.54946), final = TRUE)
$sigma.p
[1] 0

$sigma.m
[1] 0.01576683

$sigma.0
[1] 0

$logLikelihood
[1] -3.568224
\end{verbatim}}

\section{ML estimation in hierarchical SDE models, Section \ref{sec:hier}}
\label{app:hier}

Extending the function \texttt{uvector} of \ref{app:sde},
{\small\begin{verbatim}
> uvector
function(
  # If final = FALSE (default), returns vector whose sum of squares
  #   should be minimized.
  # If final = TRUE, passing the ML parameter estimates returns a
  #   list with the sigma estimates and the maximized log-likelihood.
  # Requires a transformation function y = phi(x, theta), and a
  #   function phiprime(x, theta) for the derivative dy/dx, where
  #   theta is a list containing all the parameters.
  x, t, unit = NULL,  # data and unit id vectors
  beta0, beta1, eta, eta0, x0, t0,  # SDE parameters. Some of these
  #  may be local, given as a vector of values for each observation
  lambda,  # named list of additional parameters(s) for phi,
  #  possibly local vectors
  mum = 1, mu0 = 1, mup = 1,  # optional sigma multipliers
  sorted = FALSE,  # data already ordered by increasing t?
  final = FALSE)
{
  if(is.null(unit)) unit <- rep(1, length(x))  # single unit
  if(length(unique(eta)) > 1 || length(unique(eta0)) > 1)
    stop("eta and eta0 must be global")
  theta <- data.frame(unit, beta0, beta1, eta, eta0, x0, t0,
    lambda, mum, mu0, mup)[!duplicated(unit), ]  # one row per unit
  v <- c()
  n <- logJ <- 0
  for(id in theta$unit){
    theta.j <- theta[match(id, theta$unit), ]
    j <- unit == id
    x.j <- x[j]
    t.j <- t[j]
    if(!sorted){ # ensure increasing t
      s <- order(t.j)
      x.j <- x.j[s]
      t.j <- t.j[s]
    }
    n.j <- length(x.j)
    y <- phi(x.j, theta.j)
    y0 <- phi(theta.j$x0, theta.j)
    Dt <- diff(c(theta.j$t0, t.j))
    muetam <- theta.j$mum^2 * theta.j$eta
    mueta0 <- theta.j$mu0^2 * theta.j$eta0
    muetap <- theta.j$mup^2 * (1 - theta.j$eta)
    if (theta.j$beta1 != 0) {
      ex <- exp(theta.j$beta1 * Dt)
      ex2 <- ex^2
      z <- y + theta.j$beta0 / theta.j$beta1 - ex *
        (c(y0, y[-n.j]) + theta.j$beta0 / theta.j$beta1)
      cdiag <- (ex2 + 1) * muetam + muetap * (ex2 - 1) /
        (2 * theta.j$beta1)
      cdiag[1] <- cdiag[1] - ex2[1] * (muetam - mueta0)
      csub <- -ex * muetam
    } else { # beta1 == 0
      z <- y - c(y0, y[-n.j]) - theta.j$beta0 * Dt
      cdiag <- 2 * muetam + muetap * Dt
      cdiag[1] <- cdiag[1] - muetam + mueta0
      csub <- rep(- muetam, n.j)
    }
    ld.v <- logdet.and.v(cdiag, csub, z)
    v <- c(v, ld.v$v)
    logJ <- logJ + sum(log(abs(phiprime(x.j, theta.j)))) -
      ld.v$logdet
    n <- n + n.j
  }
  if(n != length(x)) stop("Should not happen, something wrong!")
  Jn <- exp(logJ / n)   # J^(1/n)
  u <- v / Jn
  if (!final) return (u)  # "normal" exit
  # Else, at optimum, calculate sigma.P, sigma.M and sigma.Z
  #   estimates, and the log-likelihood:
  ms <- sum(u^2) / n  # mean square
  sigma2 <- Jn^2 * ms  # estimate for sigma^2
  list(sigma.P = sqrt((1 - eta) * sigma2),
       sigma.M = sqrt(eta * sigma2),
       sigma.Z = sqrt(eta0 * sigma2),
       logLikelihood = -(n / 2) * (log(ms) + log(2 * pi) + 1))
}
\end{verbatim}}
For clarity, a simple \texttt{for} loop has been used rather than potentially more efficient or elegant \textit{R}-specific shortcuts. This \texttt{uvector} is backward compatible with the previous version, giving the same results for single units.

\section{Example 3, Section \ref{sec:ex3}}
\label{app:ex3}

The transformation functions are
{\small\begin{verbatim}
> phi <- function(H, theta) with(theta,
+   ifelse(abs(rep_len(c, length(H))) < 1e-9,
+   log(H / a), ((H / a)^c - 1) / c))
> # (works for vectors and scalars)
> phiprime <- function(H, theta) with(theta,
+   (H / a)^(c - 1) / a)
\end{verbatim}}

First, take $a$ as local. Local parameters in \texttt{nls} are indexed by a factor identifying the units, \texttt{Seed} in this case.  A vector with one element per observation is passed on to the model function. On the other hand, values for the factor levels, one element per unit, are used in the starting values and output. The ML estimates are:
{\small\begin{verbatim}
> (alocal <- nls(~uvector(x = height, t = age, unit = Seed,
+ beta0 = 0, beta1 = -b, eta = 0, eta0 = 0, x0 = 0, t0 = 0,
+ lambda = list(a = a[Seed], c = c), mup = sqrt(abs(b))),
+ data = Loblolly, start = list(a = rep(72, 14), b = 0.1, c = 0.5)))
Nonlinear regression model
  model: 0 ~ uvector(x = height, t = age, unit = Seed, beta0 = 0,
       beta1 = -b, eta = 0, eta0 = 0, x0 = 0, t0 = 0, lambda =
       list(a = a[Seed], c = c), mup = sqrt(abs(b)))
   data: Loblolly
      a1       a2       a3       a4       a5       a6       a7       a8
68.36651 69.11596 71.87593 70.69002 70.44039 71.38285 72.90628 70.92199
      a9      a10      a11      a12      a13      a14        b        c
74.01902 74.77264 75.44943 76.41765 76.91871 78.84126  0.09472  0.49182
 residual sum-of-squares: 40.35

Number of iterations to convergence: 3
Achieved convergence tolerance: 5.947e-06
> p <- unname(coef(alocal))  # parameter estimates
> uvector(x = Loblolly$height, t = Loblolly$age, unit = Loblolly$Seed,
+ beta0 = 0, beta1 = - p[15], eta = 0, eta0 = 0, x0 = 0,
+ t0 = 0, lambda = list(a = p[Loblolly$Seed], c = p[16]), mup =
+ sqrt(abs(p[15])), final = TRUE)
$sigma.P
[1] 0.03358892

$sigma.M
[1] 0

$sigma.Z
[1] 0

$logLikelihood
[1] -88.39581

> logLik(alocal)  # another way
'log Lik.' -88.39581 (df=17)
> AIC(alocal)  # AIC or BIC for comparing models with different
[1] 210.7916
> BIC(alocal)  #   numbers of parameters (lower is better)
[1] 252.1155
\end{verbatim}}

Now try $a$ global and $b$ local:
{\small\begin{verbatim}
> (blocal <- nls(~uvector(x = height, t = age, unit = Seed,
+ beta0 = 0, beta1 = -b[Seed], eta = 0, eta0 = 0, x0 = 0, t0 = 0,
+ lambda = list(a = a, c = c), mup = sqrt(abs(b[Seed]))),
+ data = Loblolly, start = list(a = 72, b = rep(0.1, 14), c = 0.5)))
Nonlinear regression model
  model: 0 ~ uvector(x = height, t = age, unit = Seed, beta0 = 0,
       beta1 = -b[Seed], eta = 0, eta0 = 0, x0 = 0, t0 = 0,
       lambda = list(a = a, c = c), mup = sqrt(abs(b[Seed])))
   data: Loblolly
       a       b1       b2       b3       b4       b5       b6       b7
73.08143  0.08912  0.09082  0.09495  0.09053  0.08915  0.09111  0.09496
      b8       b9      b10      b11      b12      b13      b14        c
 0.08957  0.09680  0.09819  0.09843  0.09984  0.09984  0.10313  0.49156
 residual sum-of-squares: 37.35

Number of iterations to convergence: 4
Achieved convergence tolerance: 2.256e-06
> p <- unname(coef(blocal))
> uvector(x = Loblolly$height, t = Loblolly$age, unit = Loblolly$Seed,
+ beta0 = 0, beta1 = -(p[-1])[Loblolly$Seed], eta = 0, eta0 = 0,
+ x0 = 0, t0 = 0, lambda = list(a = p[1], c = p[16]),
+ mup = sqrt(abs((p[-1])[Loblolly$Seed])), final = TRUE)
$sigma.P
[1] 0.03231109

$sigma.M
[1] 0

$sigma.Z
[1] 0

$logLikelihood
[1] -85.15201

> c(AIC(blocal), BIC(blocal))
[1] 204.3040 245.6279
\end{verbatim}}

Finally, with both $a$ ad $b$ locals,
\begin{verbatim}
> (ablocal <- nls(~uvector(x = height, t = age, unit = Seed,
+ beta0 = 0, beta1 = - b[Seed], eta = 0, eta0 = 0, x0 = 0, t0 = 0,
+ lambda = list(a = a[Seed], c = c), mup = sqrt(abs(b[Seed]))),
+ data = Loblolly, start = list(a = rep(72, 14), b = rep(0.1, 14),
+ c = 0.5)))
Nonlinear regression model
  model: 0 ~ uvector(x = height, t = age, unit = Seed, beta0 = 0,
       beta1 = -b[Seed], eta = 0, eta0 = 0, x0 = 0, t0 = 0,
       lambda = list(a = a[Seed], c = c), mup = sqrt(abs(b[Seed])))
   data: Loblolly
      a1       a2       a3       a4       a5       a6       a7
68.31203 67.27849 68.61879 73.81971 75.66951 75.05226 72.12885
      a8       a9      a10      a11      a12      a13      a14
76.93856 72.22306 72.02841 73.72472 74.01390 75.85147 76.05177
      b1       b2       b3       b4       b5       b6       b7
 0.09497  0.09821  0.10087  0.08990  0.08659  0.08913  0.09631
      b8       b9      b10      b11      b12      b13      b14
 0.08571  0.09807  0.09972  0.09784  0.09888  0.09668  0.09960
       c 
 0.49060 
 residual sum-of-squares: 30.67

Number of iterations to convergence: 4 
Achieved convergence tolerance: 6.068e-07
> logLik(ablocal)
'log Lik.' -76.87568 (df=30)
> c(AIC(ablocal), BIC(ablocal))
[1] 213.7514 286.6759
\end{verbatim}

\section{Example 4, Section \ref{sec:ex4}}
\label{app:ex4}

Fit the $b$-local version from Example 3 as a mixed-effects model:
{\small\begin{verbatim}
> library(nlme)
> (blocal.nlme <- nlme(0 ~ uvector(x = height, t = age, unit = Seed,
+ beta0 = 0, beta1 = -b, eta = 0, eta0 = 0, x0 = 0, t0 = 0,
+ lambda = list(a = a, c = c), mup = sqrt(abs(b))), data = Loblolly,
+ fixed = a + b + c ~ 1, random = b ~ 1, groups = ~Seed,
+ start = c(a = 72, b = 0.1, c = 0.5),
+ control = nlmeControl(pnlsTol = 0.01)))
Nonlinear mixed-effects model fit by maximum likelihood
  Model: 0 ~ uvector(x = height, t = age, unit = Seed, beta0 = 0,
         beta1 = -b, eta = 0, eta0 = 0, x0 = 0, t0 = 0,
         lambda = list(a = a, c = c), mup = sqrt(abs(b)))
  Data: Loblolly
  Log-likelihood: -101.7804
  Fixed: a + b + c ~ 1
          a           b           c
73.43276915  0.09381183  0.49381270

Random effects:
 Formula: b ~ 1 | Seed
                  b  Residual
StdDev: 0.003814812 0.7307142

Number of Observations: 84
Number of Groups: 14
\end{verbatim}}
To achieve convergence it was necessary to increase the tolerance \texttt{pnlsTol} from the default 0.001. It may be wise to scale the variables, e.~g., specifying \texttt{x = height / 10, t = age / 10} and adjusting the \texttt{start} values accordingly, to make quantities closer to 1 as recommended in the \texttt{nlme} documentation. That did not make much difference in this instance.

One can obtain \emph{predictions} for the units (not \emph{estimates}, since here the locals are random variables and not parameters):
{\small\begin{verbatim}
> coef(blocal.nlme)
           a          b         c
329 73.43277 0.08972720 0.4938127
327 73.43277 0.09095265 0.4938127
325 73.43277 0.09400202 0.4938127
307 73.43277 0.09087038 0.4938127
331 73.43277 0.08972664 0.4938127
311 73.43277 0.09129559 0.4938127
315 73.43277 0.09407617 0.4938127
321 73.43277 0.09010012 0.4938127
319 73.43277 0.09536543 0.4938127
301 73.43277 0.09631403 0.4938127
323 73.43277 0.09647272 0.4938127
309 73.43277 0.09739088 0.4938127
303 73.43277 0.09743172 0.4938127
305 73.43277 0.09964012 0.4938127
>
> c(AIC(blocal.nlme), BIC(blocal.nlme))
[1] 213.5608 225.7149
> c(AIC(blocal), BIC(blocal))
[1] 204.3040 245.6279
\end{verbatim}}

\bibliographystyle{spbasic}      
\bibliography{sdes3}   

\addcontentsline{toc}{section}{References}

\end{document}